\documentclass[atmp]{ipart_v1}

\Vol{?}
\Issue{?}
\Year{?}
\firstpage{1}

\usepackage{t1enc}
\usepackage[latin1]{inputenc}
\usepackage[english]{babel}
\usepackage{color}
\usepackage{amsthm,amsmath,amssymb,amscd}
\usepackage{amssymb}
\usepackage{graphics}

\usepackage{slashed}

{

\newtheorem{theorem}{Theorem}[section]

\newtheorem{definition}{Definition}[section]
\newtheorem{remark}{Remark}[section]
\newtheorem{proposition}{Proposition}[section]

\newtheorem{lemma}{Lemma}[section]

\newcommand{\RR}{{\mathbb{R}}}

\newcommand{\CC}{{\mathbb{C}}}

\newcommand{\p}{\partial}

\newcommand{\pa}{\left(\frac{\p}{\p z}\right)^a}

\newcommand{\pab}{\left(\frac{\p}{\p \bar z}\right)^a}

\newcommand{\beq}{\begin{eqnarray}}
\newcommand{\eeq}{\end{eqnarray}}

\newcommand{\R}{\mathbb{R}}
\newcommand{\N}{\mathbb{N}}

\begin{document}
\title[Fractional Virasoro Algebra]{ Fractional Virasoro Algebras}

\author{Gabriele La Nave and Philip W. Phillips}

\begin{abstract}
We show that it is possible to construct a Virasoro algebra as a central extension of the fractional Witt algebra generated by non-local operators of the form, $L_n^a\equiv\left(\frac{\p f}{\p  z}\right)^a$ where $a\in {\mathbb R}$.   The Virasoro algebra is explicitly of the form,
\beq
[L^a_m,L_n^a]=A_{m,n}(s)\otimes L^a_{m+n}+\delta_{m,n}h(n)cZ^a
\eeq
where $A_{m,n}(s)$ is a specific meromorphic function $c$ is the central charge (not necessarily a constant), $Z^a$ is in the center of the algebra and $h(n)$ obeys a recursion relation related to the coefficients $A_{m,n}$. In fact, we show that all central extensions which respect the special structure developed here which we term a multimodule Lie-Algebra, are of this form. This result provides a mathematical foundation for non-local conformal field theories, in particular recent proposals in condensed matter in which the current has an anomalous dimension.
 \end{abstract}
\maketitle

\section{Introduction}

The Virasoro\cite{virasoro1,fubini} algebra is central to string theory as it underpins the conformal structure of the local current operators.   In one of its incarnations, it is constructed as the central extension of the Witt algebra as the space of {\it local} conformal transformations on the unit disk. 
Consequently, for any problem controlled by critical scaling, the Virasoro algebra governing the conserved currents is of fundamental importance. In all constructions of the Virasoro algebra thus far, the generators are entirely local.  However, there are a number of physical problems in which the currents are inherently non-local and hence require a fundamentally new Virasoro algebra.  Consider, for example, three-dimensional bosonization\cite{schaposnik} of massless fermions which results in a bosonic non-local Maxwell-Chern-Simons theory in which the kinetic energy operator is the fractional Laplacian, in particular $\Box^{\frac12}$.  Another example is found in the AdS/CFT correspondence\cite{bdhm,Witten1998,klebanov,polchinski}  in which the boundary operator dual to a bulk massive scalar field is again the fractional Laplacian\cite{glp2015}, where the power is determined by the mass of the scalar field.  A possible application of this\cite{glp2015} is the the strange metal in the cuprate high-temperature superconductors.  This problem has long been argued to be controlled by quantum critical scaling\cite{Valla1999,anderson,Marel2003}.  Recent phenomenology\cite{hk2014} on this problem is based on a vector potential that has an anomalous dimension.   Indeed, an anomalous dimension for the vector potential is problematic because the local gauge symmetry of electricity and magnetism, $A_\mu\rightarrow A_\mu+\partial_\mu\Lambda$, requires that $[A_\mu]=1$.  While an anomalous dimension for the vector potential can emerge in bulk Lifshitz theories\cite{kiritsis}, this does not solve the problem of how to maintain gauge invariance in the resultant boundary theory.   Regardless of the underlying quantum field theory, squaring an anomalous dimension of the vector potential with gauge invariance necessitates (assuming only that the underlying symmetry is still based on an infinitesimal transformation) a new non-local symmetry\cite{fracdiff,GWPW} of the form 
\beq
A_\mu\rightarrow A_\mu+d_a\Lambda,\quad [A_\mu]=a
\eeq
where $d_a=(\Delta)^{(a-1)/2} d$, $d$ the complete exterior derivative and $a\in {\mathcal R}$.   Such a transformation leads to non-local currents\cite{lp2015, GWPW}, and hence if there is a string formulation of this problem, a Virasoro algebra allowing arbitrary fractional dimensions of the currents must be formulated.  It is this task that we perform here.  Indeed, fractional generalizations\cite{zamaf1983,argyres1993}of the Virasoro algebra do exist which do allow central charges that exceed unity.  However, such generalizations are not directly applicable to the strange metal as they all have currents that transform without an anomalous dimension.  

To solve this problem and the general phenomena of non-local currents,  we construct a family of Lie algebras (which are modules over a certain commutative Lie algebra $\mathcal H$ of holomorphic functions over $\mathbb C$) $\mathcal V_a$, thereby generalizing the Virasoro algebra. These algebras $\mathcal V_a$ are defined as central extensions of the algebra $\mathcal W_a$ (itself a generalization of the Witt algebra) which consists of operators which are combinations -- linear over $\mathcal H$-- of the form: $\sum _n \phi _n(s) L_n^a$, where $\phi (s)\in \mathcal H$ and $L_n^a\equiv\left(\frac{\p f}{\p  z}\right)^a$ is the fractional $z$-derivative. We think of these operators as acting on ``power series'' $\sum _k p_k z^{ak}$ via
\beq
\left(\sum _n \phi _n(s) L_n^a\right) \left(  \sum _k p_k z^{ak} \right)= \sum _k p_k \,\phi _n(ak) \,  L_n^a ( z^{ak} ).
\eeq
This could be all reformulated in a more analytical and invariant way, but we choose not to because we mean to focus on the algebraic structure of the algebras.   

The algebra $\mathcal W_a$ has a special structure, which we call Lie multimodule in Definition \ref{multimodule}; namely, there are operations $\star _{(p,q)}$ on $\mathcal H$ and a grading on $\mathcal W_a$, such that: $[\phi\otimes L_p, \psi \otimes L_q]= \phi \star _{p,q} \psi[L_p,L_q]$ (cf. eq. \eqref{eq-star} for the definition of $\star_{(p,q)}$). 
 We show that all the central extensions which preserve this extra structure
\beq
0\to \mathcal H \to \mathcal V _a \to \mathcal W_a\to 0,
\eeq
which are parametrized by a group $H^2_\star(\mathcal W_a,\mathcal H)$ (which we show to be isomorphic to $\mathcal H$) are of the form 
\beq
[L^a_m,L_n^a]=A_{m,n}L^a_{m+n}+\delta_{m,n}h(n)cZ^a
\eeq
where $c$ is the central charge ($c\in \mathcal H$), $Z^a$ is in the center of the algebra and $h(n)$ obeys the recursion relation,
\begin{equation} \begin{aligned}& h(2)=c \\&
\frac{A_{-1, -m} \Gamma _{(-(m+1))}- A_{m,1} \Gamma _{m+1}}{A_{-(m+1), m+1}}  \,  \,h((m+1)) \\=& \frac{A_{m+1, -1} \Gamma _1 
-A_{1,-(m+1)} \Gamma _{-m}}{A_{m,-m}}  \,h(m)\\\end{aligned}
\end{equation}

Here
\beq A_{p,q}(s)=  \frac{\Gamma (a(s+p)+1)}{\Gamma (a(s-1+p)+1)} - \frac{\Gamma (a(s+q)+1)}{\Gamma (a(s-1+q)+1)}    \eeq
and 
\beq \Gamma _p(s)= \frac{\Gamma (a(s+p)+1)}{\Gamma (a(s-1+p)+1)}\eeq
where $\Gamma$ is the gamma function. 
The elements of $\mathcal W_a$ are operators acting on $\mathbb C [[z^{a}, z^{-a}]]$ via the prescription
\beq \left(\phi \otimes L_p^a\right) (z^{ka})= \phi(k)L_p^a (z^{ka}).\eeq

\noindent
The usual $\mathcal H$-Lie algebras $\mathcal V _a$ are a generalization of the Virasoro algebra in that
\beq
\lim _{a\to 1} \mathcal V_a= V.
\eeq

The Lie algebra structure of $\mathcal V_a$, on the other hand,  does not arise, for $a\neq 1$, as a tensor product of a Lie algebra $V$ with $\mathcal H$, reflecting the very non-local nature of the operators in $\mathcal V_a$. In this sense, it is a twisted structure, or more properly a Lie multimodule, further indicating the non local nature of non-local conformal field theories.

\section{Fractional (holomorphic) Derivatives}

\subsection{ Holomorphic and anti-holomorphic fractional derivatives}
 
 In this section we intend to emphasize that the type of operators we consider have various nice analytic incarnations. Nonetheless, we ultimately take a purely algebraic approach to the definition of the operators and the vector spaces (see the next section) they act upon. 
 Recall that the fractional Laplacian in $\RR ^n$ can be defined as (a regularization of)
 $$(-\Delta_x)^a f(x)=C_{n,a}\int_{\R^n}\frac{f(x)-f(\xi)}{\mid {x-\xi}\mid^{n+2a}}\;d\xi$$
for some constant $C_{n, a}$ or equivalently in terms of its Fourier transform and hence as a {\it pseudo-differential} operator.
Given a complex valued function $f$, let us denote by $\hat f$ its Fourier transform. There are various ways to generalize the concept of the holomorphic derivative.
One possibility is to consider the following approach analogous to the one used for the fractional Laplacian.
\begin{definition}
The fractional holomorphic derivative is\begin{footnote}{ This is the same as defining $\pa f = \frac{\partial}{\partial z} (-\Delta )^{\frac{a-1}{2}} f$ and the analogous expression for $ \left(\frac{\p f}{\p  \bar z}\right)^a $}\end{footnote}
\beq \hat\pa f =  c_a \xi^{a} \hat f,\eeq
where $\xi= \xi_1+ i \, \xi _2$ is the complex momentum and analogously the fractional antiholomorphic derivative is
\beq\widehat{\left(\frac{\p f}{\p  \bar z}\right)^a } = \bar c_a \bar \xi ^{a} \hat f.
\eeq
\end{definition}
As a result, 
$$\widehat { \pa \left( \left(\frac{\p f}{\p  \bar z}\right)^a\right)}= |c_a|^2 \, |\xi|^{2a} \, \hat f$$
or equivalently, we have shown the following Lemma.
\begin{lemma}
The fractional holomorphic and antiholomorphic derivatives are such that
$$\pa \pab f = |c_a|^2\, (-\Delta )^a f,$$
where $(-\Delta )^a$ is the Riesz fractional Laplacian. \end{lemma}

One could also take a slightly different tack and consider instead the Louiville approach to fractional calculus as our starting point, which in turn has its origins in the classical Cauchy integral formula for the holomorphic derivative of analytic functions. One could then define,
$$\pa f (z)=\frac{\Gamma(a+1)}{2\pi i}\, \int _\gamma \; \frac{f(\xi)}{(\xi - z)^{(1+a)}}\, d\xi,$$
for any loop $\gamma$ around $z$ (at any point on the cut, i.e., the negative real axis, the loop must be thought of as lifted on the universal covering). This definition takes into account that the kernel of the integral, namely $\frac{1}{(\xi - z)^{(1+a)}}$, has now a branch rather than a pole at the origin, and therefore one must consider introducing a cut in the complex plane.  For example, removing the non-positive semi-axis along the real line will do.
These pseudo-differential operators will be the building blocks for the fractional Witt algebra.

\subsection{Fractional derivatives: algebraic formulation}
We consider the universal cover $\mu:R\to \CC^*$ where $\mu (z) = e^z$. It is well known that $R\simeq \CC$. For $a\in \RR$,  $w^a$ defines a map from $\CC ^*$ to itself through $w^a= e^{ax} e^{iay}$, where $w=e^z$ for $z= x+iy$. The map $w\mapsto w^a$ is defined on $\CC^*$ and it is covered by the homothety of $R$: $z\mapsto az $.
From now on we consider $z^a$ as a symbolic expression, with the understanding that it has the incarnation described above. Consider the algebra,
$$V^a:=\CC[ [z^{-a}, z^a]],$$
of formal power series in $z^a$ and $z^{-a}$.
We denote by $V^a_k$, the $\CC$-vector space spanned by $z^{ak}$ and we define the operator $\pa$ on $V_k^a$ as the linear operator,
$$\pa : V_k^a\to V_{k-1}^a,$$
defined by (for $ak>-1$)\begin{footnote}{The choice of $A_k=1$ that we make in this note clashes a bit with the analytic definition we gave in the previous section, depending on the value of $a$, but it is easy to adjust the rest of the paper with an arbitrary $A_k$ in the rest of the paper. }\end{footnote}
\begin{equation} \label{holofrac-def}\pa (z^{ak}) =A_k c_a \, \frac{\Gamma (ak+1)}{\Gamma (a(k-1)+1)}\, z^{a(k-1)},\end{equation}
and
$$\pa (\bar z^{ak})=0,$$
where $\Gamma$ is the Gamma function, $c_a$ is an unspecified constant and $A_k$ only depends on $k$ (which in this note, we will take to equal $1$ for notational convenience). After the change of coordinates, $\zeta= z^a$, it is apparent that the action of $\pa$ on $V_k^a$ for $k\neq 0$ is equivalent to the action of
$$c_a \, \frac{\Gamma (ak+1)}{\Gamma (a(k-1)+1)}\, \frac{1}{k}\frac{\p}{\p \zeta},$$
on the degree $k-$subspace (which we denote by $V_k$) of $V:=\CC [[\zeta ^{-1} , \zeta]]$. 
The operator $\pa$ thus defined is definitely not a derivation because it fails the Leibnitz rule, as in fact $$\pa (1)=A_0 c_a \frac{ z^{-a}}{\Gamma (1-a)},$$ which is non-zero unless $a\in \N$ (or $A_0=0$).
In the rest, for the sake of notation, we will assume $A_k=1$ for every $k$, and thus take
\beq \label{oplus}\pa= \bigoplus _k c_a \, \frac{\Gamma (ak+1)}{\Gamma (a(k-1)+1)}\, P_k\eeq
where $P_k: V_k \to V_{k-1}$ is defines as $P_k(z^{ak}))= z^{a(k-1)},$ or equivalently, after the change of coordinates, $\zeta= z^a$, 
\beq \label{zetaoplus} \pa= \bigoplus _k c_a \, \frac{\Gamma (ak+1)}{\Gamma (a(k-1)+1)}\, \frac{1}{k}\frac{\p}{\p \zeta}.\eeq

\noindent
thus showing that after the change of variables  $\zeta= z^a$, the operator $\pa$ does not equal $\frac{1}{k}\frac{\p}{\p \zeta}$, even up to multiples.

\begin{remark}
We would like to emphasize that the actual form of the coefficients in the defining equation \eqref{holofrac-def} is not important, we could choose different coefficients. That is we could define:
$\pa (z^{ak}) = C_{a,k} \,  z^{a(k-1)}$ and have a similar algebra. The constructions to follow will all go through, {\it mutatis mutandis}.
\end{remark}


\section{Fractional Virasoro Algebra}
\subsection{Lie multimodules}
In this section we define the notion of a Lie multimodule over an algebra. Let $\mathcal A$ be an algebra endowed with a family of operations: $\star _{p,q}$ parametrized by $p,q\in \mathbb Z$. 
\begin{definition} \label{multimodule}A Lie {\bf multi-module} over $(\mathcal  A,[\cdot, \cdot]_{\mathcal A,p,q})$ (or $\mathcal A$-Lie algebra for short) is a graded $\mathcal A$-module $\mathcal W=\oplus _k \mathcal W_k$ endowed with a Lie bracket
$$[\, \cdot \, ,\,  \cdot\, ] : \mathcal W \times \mathcal W\to \mathcal W,$$
such that
$$[a_1 v, a_2w]= a_1\star _{p,q} a_2 [v,w],$$
for any $v\in  \mathcal W_p, \; w\in  \mathcal W_q.$
\end{definition}
A {\it trivial} example of such a structure is the standard Lie-module over a Lie algebra. For instance, the trivial one obtained by tensoring a Lie algebra $(V, [\,\cdot\,, \,\; \,])$ with a commutative Lie algebra $\mathcal A$ where the Lie bracket of $V\otimes \mathcal A$ is given by
\begin{equation}\label{tensorbracket} [v\otimes \phi, w\otimes \psi]= \phi \psi \,[v,w]\end{equation}
for every $v, w\in V$ and $\phi,\psi\in \mathcal A$ is such an example.
In the next section we will construct a more complex example of such a structure.
\subsection{The fractional Witt algebra}
We now define the infinite dimensional Lie algebra of pseudo-differential operators which is a generalization of the standard Witt algebra as follows.
We consider
\begin{equation}
L^a_n= -z^{a(n+1)} \pa ,\qquad {\bar L } ^a_n:=-\bar z^{a(n+1)} \pab
\end{equation}
acting on $V^a:=\CC[ [z^{-a}, z^a]]$. In this algebraic description we think of $z^a$ merely as a formal expression as in Puiseaux series, if $a\in \mathbb Q$ \cite{e2000}.

For any integer $p$, we define the following functions
\begin{equation}\label{Gammadef} \Gamma _p(s) := \frac{\Gamma (a(s+p)+1)} {\Gamma (a(s+p-1)+1) }\end{equation}
and
\begin{equation}\label{A} A_{p,q}(s)= \Gamma _p(s) - \Gamma _q(s)=\left( \frac{\Gamma (a(s+p)+1)}{\Gamma (a(s-1+p)+1)} - \frac{\Gamma (a(s+q)+1)}{\Gamma (a(s-1+q)+1)}    \right). \end{equation}
Clearly
$$\Gamma _p (s)= \Gamma _0(s+p).$$
Let $ \mathcal M (\mathbb C)$ be the algebra of meromorphic functions on $\mathbb C$ (when needed, we will denote by $ \mathcal M_k (\mathbb C)$ the set of meromorphic functions which only have poles of order k) and set{\begin{footnote} {We could actually take functions in $\mathcal F$ to have more regularity, by taking $\mathcal F:= \mathcal M(\mathbb C) \cap \mathcal H (\mathbb C\setminus S)$ where $S:= \{z\in \mathbb C: \; Re(z) \in -\frac{1}{a} \mathbb N-\mathbb N, \; Im(z)=0\}$ and $H (\mathbb C\setminus S)$ is the space of holomorphic functions on $\mathbb C\setminus S$}  \end{footnote}}
$$\mathcal F:=\left \{ f\in \mathcal M (\mathbb C):\; f (z) \text{ is holomorphic in a neighborhood of } z \text{ if } Re(z)\in \mathbb Z\right\},$$
which we think of as a sub-algebra of meromorphic functions on $\mathbb C$.
We endow $\mathcal F$ with the family of operations

\beq \label{eq-star}\phi\star _{p,q} \psi:= \frac{\psi (p+s) \phi(s) \,\Gamma _p - \phi(q+s) \psi(s) \Gamma _q}{A_{p,q}}\eeq

and the related family of "brackets"
\beq\; [\phi(s), \psi(s)]_{\mathcal H, p,q}=\left\{ \begin{aligned} &  \psi (p+s) \phi(s) \,\Gamma _p - \phi(q+s) \psi(s) \Gamma _q \text{ for } \; p\neq q\\& 0  \;\text{ for } \; p= q.
\end{aligned} \right.\eeq
 \begin{definition}
 We define the algebra $\mathcal H$ to be the subalgebra of $\mathcal F$ generated by $\CC$ closed under  $ [\phi(s), \psi(s)]_{\mathcal H, p,q}= \psi (p+s) \phi(s) \,\Gamma _p - \phi(q+s) \psi(s) \Gamma _q$, for every $p,q \in \mathbb Z$.
 \end{definition}
 Clearly $\mathcal H$ is contained in the algebra $\mathbb C [A_{p,q}, \Gamma _\ell]_{p,q,\ell\in \mathbb Z}$ of polynomials in $A_{p,q}, \Gamma _\ell$, but it is generally smaller. In fact, we will see that for $a=1$, $\mathcal H=\mathbb C$, thus highlighting, at the same time, the $a$-dependence of $\mathcal H$ .
 
 We then define the algebra
$$\mathcal W _a := \bigoplus_{n\in \mathbb Z}\;  \mathcal H \otimes L_n,$$
evidently an $\mathcal H$-module, which we think of as acting on $V^a:=\CC[ [z^{-a}, z^a]]$ by considering the action of a typical generator on a basis for $V^a$ via
$$\left(\phi(w)\otimes L_n\right) (z^{ak})= \phi (k) L_n(z^{ak}),$$
for $k\neq 0$. This is clearly a representation by construction; that is $\phi(w)\otimes L_n$ acts as a linear operator on $V^a$ and the bracket of elements of $\mathcal W_a$ is defined as the commutator of endomorphisms of  $V^a$.

\begin{remark}
Let us emphasize that, because of equation \eqref{oplus} or equivalently eq. \eqref{zetaoplus}, the fractional Witt algebra thus constructed (and therefore the fractional Virasoro algebra) are not isomorphic to the Witt algebra.
\end{remark}
\noindent
 A consequence of  Eq. \eqref{holofrac-def} is the following Lemma.
\begin{lemma}
The $\mathcal H$-module $\mathcal W_a$ spanned by the pseudodifferential operators $L_n$ is an $\mathcal H$-Lie algebra, when endowed with brackets $[\cdot, \cdot]$ consisting of commutators. In fact, one has
\begin{equation}\label{brackets}
[L_n, L_m] f (z) = \sum _k a_k A_{n,m}^a(k)  \, L_{n+m} (z^k)= ( A_{n,m}^a(s) \otimes L_{n+m}) (f(z))
\end{equation}
for any $f (z)= \sum _k a _k z^k$, where $ A_{n,m}^a(k)$ is the evaluation at $k$ of the meromorphic function $ A_{n,m}^a(s)$ defined in \eqref{A}
clearly $ A_{n,m}^a(s)\in \mathcal H$, if $a\notin \mathbb Z$. More generally,
\begin{equation}\label{brackets-gen} [\phi \otimes L_n, \psi \otimes L_m]=  \left( \psi(n+s) \phi(s) \Gamma _n^a - \psi (s) \phi(m+s)\Gamma _m^a\right) \otimes L_{m+n}\end{equation}.
\end{lemma}
\begin{proof}
First observe that 
\begin{equation}\label{L-value}
L_p ( z^{a\ell})= -\Gamma _0(\ell) \, z^{a(\ell+p)},
\end{equation}
as one can verify via
$$L_p( z^{a\ell})= - z^{a(p+1)} \pa( z^{a\ell})=- z^{a(p+1)}  \frac{\Gamma (a\ell+1)}{\Gamma (a(\ell-1)+1)}\, z^{a(\ell-1 )} = -\Gamma _0(\ell) \, z^{a(\ell+p)} $$
Using this, we find that
$$\begin{aligned} L_n( L_m \, z^{ak})&=   z^{a(n+1)} \pa \left( \Gamma _0(k)\, z^{a(k+m)} \right) 
= \frac{\Gamma (a(k+m)+1)}{\Gamma (a(k-1+m)+1)}  \;\Gamma _0(k) z^{a(k+m+n)}\\&= \Gamma _m(k) \, L_{n+m} (z^{ak} )=\left(\Gamma _m(s)\otimes   L_{n+m}(z^{ak} )\right),\end{aligned}$$
whence
\begin{equation}
L_n\circ L_m =\Gamma _m(s)\otimes   L_{n+m},
\end{equation}
where $\circ$ denotes composition.
%
From this (using \eqref{L-value} again), we obtain

$$\begin{aligned}  (\phi \otimes L_n) (\psi \otimes L_m) (z^{ak})&=(\phi \otimes L_n) \left(- \psi(k)   \Gamma _0(k) \, z^{a(k+m)}  \right) \\&=\phi(m+k) \psi (k) \Gamma _m(k) _m\Gamma _0(k) \, z^{a(k+m+n)}\\&= -\phi(m+k) \psi (k)  \Gamma _m(k) L_{n+m}  (z^{ak}). \end{aligned}$$
Thus,
\begin{equation}
 (\phi \otimes L_n) \circ (\psi \otimes L_m) = -\left(\phi(m+s) \psi (s)  \Gamma _m(s) \right) \otimes L_{n+m} .
\end{equation}
As a result, we have
\begin{equation}
 [\phi \otimes L_n, \psi \otimes L_m]=  \left( \psi(n+s) \phi(s) \Gamma _n^a - \psi (s) \phi(m+s)\Gamma _m^a\right) \otimes L_{m+n},
\end{equation}
which is Eq. \eqref{brackets-gen}. From this, Eq. \eqref{brackets} follows directly, taking $\phi=\psi=1$.

Clearly the properties of the $\Gamma$-function imply that $ A_{n,m}^a(s)\in \mathcal H$. 
\end{proof}

We next intend to show that indeed the Lie algebra structure is simple for $a=1$.
\begin{lemma}
One has
$$\lim _{a\to 1} \mathcal W_a= W .$$

\end{lemma}

\begin{proof}

The $\mathcal H$-Lie algebra structure of $\mathcal W_a$ is determined by the formula in Eq. \eqref{brackets}
$$[L_n, L_m] \phi (z) = \sum _k a_k A_{n,m}^a(k)  \, L_{n+m} (z^k),$$
where the holomorphic functions $A_{n,m}^a(s)$ are given by Eq. \eqref{A}. We next observe that, since
$$\Gamma (z+1)= z\, \Gamma(z),$$
for $a=1$ and any $s$

$$\Gamma _p(s)= \frac{\Gamma ((s+n)+1)}{\Gamma ((s-1+n)+1)} = \frac{(s+n) \Gamma (s+n)}{ \Gamma (s+n)}= s+n.$$

Using this computation,  we can show that $\mathcal H=\mathbb C$, when $a=1$. In fact, by definition, $\mathcal H$ is the algebra generated by $\mathbb C$, closed with respect to
$$ \phi\star _{p,q} \psi:= \psi (p+s) \phi(s) \,\Gamma _p - \phi(q+s) \psi(s) \Gamma _q.$$
A straightforward computation, using the formula we obtained for $\Gamma_p$ then yields
\begin{equation} \phi\star _{p,q} \psi=\psi (p+s) \phi(s) \,(s+p)- \psi(q+s) \phi(s) (s+q) .\end{equation}
\noindent
Since $A_{n,m}^1(s)$ is $s$-independent, it follows that for $a=1$,
$$[L_n, L_m] \phi (z) =(n-m) \, L_{n+m} (\phi(z)),$$
which is the standard structure of the Witt algebra. 
Also for elements in $\mathcal H$, by definition of the family of Lie brackets and the computations above, for $p\neq q,$
$$[\phi(s), \psi(s)]_{\mathcal H, p,q}=\frac{ \psi (p+s) \phi(s) \,\Gamma _p - \psi(q+s) \phi(s) \Gamma _q}{A_{p,q}}= \psi (s) \phi(s) \frac{\Gamma _p -  \Gamma _q}{A_{p,q}}=\psi (s) \phi(s), $$
and for $p=q$
$$[\phi(s), \psi(s)]_{\mathcal H, p,p}=0,$$
since $\mathcal H$ is the smallest algebra which is a Lie algebra for every $[\cdot, \cdot]_{\mathcal H, p,q}$ and which contains $\mathbb C$. 
\end{proof}
\begin{remark}
We make the observation that the $\mathcal H$-Lie algebra structure is necessary, and that in general it does not arise from a tensor product of a Lie algebra over $\mathcal C$ (a standard Lie algebra) and the commutative Lie algebra $\mathcal H$, unless $a=1$. This can be seen from the fact that $A_{n,m}^a(k)$ is in general $k$-dependent.
\end{remark}
Here we record some useful identities about the functions $A_{m,n}(s)$.
\begin{lemma}

\begin{itemize}
\item For every $m,n$, $A_{m,n}(s)=A_{m,0}(s)+A_{0,n}(s)$.
\item For  every $m$, $A_{m,0} (s)= A_{0,-m}(s+m)$.
\item One has \begin{equation}\label{jacobiLie}\begin{aligned} &A_{m+\ell,n+\ell}(s)A_{\ell,0}(s)- A_{m,n}(s)A_{m+n,0}(s) +A_{n+m,\ell+m}(s)A_{m,0}(s)- A_{n,\ell}(s)A_{n+\ell,0}(s) \\&+A_{\ell+n,m+n}(s)A_{n,0}(s)- A_{\ell,m}(s)A_{\ell+m,0}(s) \\
 &+\Gamma_0(s) \left( A_{m,n}(s)-  A_{m+\ell,n+\ell}(s)+ A_{n,\ell}(s)-  A_{n+m,\ell+m}+A_{\ell,m}(s)-  A_{\ell+n,m+n}(s)\right)=0. \end{aligned}
 \end{equation}
 
\end{itemize}
\end{lemma}
\begin{proof}
The first two follow readily from the definition of $A_{m,n}(s)$. The last one is simply a reformulation of the Jacobi identity,
\begin{equation} \label{jacobioperator}\left( [L_\ell, [L_m,L_n]]+ [L_m, [L_n, L_\ell]]+ [L_n, [L_\ell, L_m]]\right) z^{as}=0.
\end{equation}
On the other hand, applying the operator $[L_r, [L_p,L_q]]$ to $z^{as}$ leads to the expression

$$\begin{aligned} &[L_r, [L_p,L_q]]\, z^{as} =  L_r \left( A_{p,q}(s) L_{p+q}(z^{as})\right) -  A_{p,q}\otimes L_{p+q} \left( L_r(z^{as} )\right)\\
&=A_{p,q}(s) L_r\left( -\Gamma_0(s) \, z^{a(p+q+s)}\right)- A_{p,q}\otimes L_{p+q} \left( -\Gamma_0(s) \, z^{a(r+s)}\right)\\&=A_{p,q}(s) \left( \Gamma_0(s) \frac{\Gamma(a(p+q+s)+1)}{\Gamma(a(p+q+s-1)+1)} \, z^{a(r+p+q+s)}\right)\\&- A_{p,q}(s+r) \left( \Gamma_0(s)\frac{\Gamma(a(r+s)+1)}{\Gamma(a(r+s-1)+1)}  \, z^{a(r+s+p+q)}\right)\\&=- A_{p,q}(s)\frac{\Gamma(a(p+q+s)+1)}{\Gamma(a(p+q+s-1)+1)} L_{p+q+r}(z^{as})+ A_{p,q}(s+r) \Gamma_0 (s)  L_{p+q+r}(z^{as}),
 \end{aligned}$$
 whence
 \begin{equation}\label{triplecommutator}
\begin{aligned} & [L_r, [L_p,L_q]]\, z^{as}=\\&\left(  A_{p+r,q+r}(s)A_{r,0}(s)- A_{p,q}(s)A_{p+q,0}(s) +\Gamma _0(s)  \left(A_{p+r,q+r}(s)-A_{p,q}(s)  \right)\right)   \, L_{p+q+r}(z^{as}). \end{aligned}
 \end{equation}
 Hence, recalling Eq. \eqref{jacobioperator}, in light of the identity in Eq. \eqref{triplecommutator}, one has
  
$$\begin{aligned} & 0=A_{m+\ell,n+\ell}(s)A_{\ell,0}(s)- A_{m,n}(s)A_{m+n,0}(s) +\Gamma_0(s)  \left( A_{m,n}(s)-  A_{m+\ell,n+\ell}(s)\right)\\
 &+A_{n+m,\ell+m}(s)A_{m,0}(s)- A_{n,\ell}(s)A_{n+\ell,0}(s) +\Gamma_0(s)  \left( A_{n,\ell}(s)-  A_{n+m,\ell+m}(s)\right) \\&+
A_{\ell+n,m+n}(s)A_{n,0}(s)- A_{\ell,m}(s)A_{\ell+m,0}(s) +\Gamma_0(s)  \left( A_{\ell,m}(s)-  A_{\ell+n,m+n}(s)\right).
 \end{aligned}
  $$
 Thus the identity of Eq. \eqref{jacobiLie} must hold.
 
%
\end{proof}

\subsection{Central extensions and cohomology}
We now consider the question of central extensions $$0\to \mathcal H \to \mathcal V _a \to \mathcal W_a\to 0,$$
which preserve the multimodule structure as in Definition \ref{multimodule}.
We will show in this section that these are parametrized by the cohomology group,
 $$H^2_\star(\mathcal W_a,\mathcal H)= Z^2_\star(\mathcal W_a,\mathcal H)/B^2_\star(\mathcal W_a,\mathcal H).$$
  We note that the Lie algebra structure of $\mathcal H$ is given by $[\phi,\psi]=0$ consistently with the fact that we identify $\phi$ and $\psi$, respectively with $\phi \otimes 1$ and $\psi \otimes 1$ so that $[\phi,\psi]= [\phi\otimes 1,\psi\otimes 1]=[\phi,\psi]_{\mathcal H, 0,0}=0$ since $1$ is degree $0$.
  Given a bilinear map,
  $$\omega : \mathcal W _a \times \mathcal W _a \to \mathcal H$$
  we say that is is {\it $\star_{p,q}$-bilinear} if
  \beq \label{starlinear}
  \omega (\phi \otimes L_p , \psi \otimes L_q)= \phi \star_{p,q} \psi \omega (L_p, L_q).\eeq
 Here the set of {\it 2-cycles} is by definition
 $$Z^2_\star(\mathcal W_a,\mathcal H)=\left \{ \omega : \mathcal W _a \times \mathcal W _a \to \mathcal H:\; \begin{aligned}  &(1)\;\omega \text{ is } \star_{p,q}\text{-bilinear}\\ &(2)\;   \omega \text{ is alternate}\\&(3)\; \omega (g_1,[g_2,g_3])+  \omega (g_2,[g_3,g_1])+ \omega (g_3,[g_1,g_2])=0\end{aligned}\right\}.$$
 Property (3) above is called the Jacobi identity.
 
 The subgroup $B^2(\mathcal W_a,\mathcal H)$ of $Z^2(\mathcal W_a,\mathcal H)$ is defined as the image via the {\it co-boundary} map $\delta$ of $Z^1(\mathcal W_a,\mathcal H)$,
 $$B^1_\star(\mathcal W_a,\mathcal H):= \{ \mathcal H\text{-linear maps} \; \lambda : \mathcal W_a \to \mathcal H\},$$
 and $\delta (\lambda)$ is defined as the ~2-chain,
 $$\delta (\lambda)(g_1, g_2)= \lambda ([g_1,g_2]).$$
Clearly, because of the Jacobi identity, which reads (cf. Eq. \eqref{jacobioperator}),
$$ [L_\ell, [L_m,L_n]]+ [L_m, [L_n, L_\ell]]+ [L_n, [L_\ell, L_m]],$$

\noindent
 one has that $\delta (\lambda)$ satisfies identity $(3)$ in the definition of $Z^2_\star(\mathcal W_a,\mathcal H)$.
We will prove the following generalization in our context.
\begin{theorem}
Let $\mathcal W_a$ be a Lie multimodule over $\mathcal H$ (in the sense of Definition \ref{multimodule}) endowed with $[\cdot, \cdot]_{\mathcal H, p,q}$. 
Isomorphism classes of central extensions of the Lie multimodule $\mathcal W_a$ by $\mathcal H$ which are still multimodules are in one-to-one correspondence with $H^2_\star(\mathcal W_a,\mathcal H)= Z^2_\star(\mathcal W_a,\mathcal H)/B^1_\star(\mathcal W_a,\mathcal H)$.
\end{theorem}
\begin{proof}
Standard theory about central extensions yields that these are parametrized by $H^2(\mathcal W _a, \mathcal H)$. In fact, it is a standard fact that the central extension
$$0\to \mathcal H\to \mathcal V_a \to \mathcal W_a\to 0$$
is equivalent to providing a bilinear, alternate map $ \omega : \mathcal W _a \times \mathcal W _a \to \mathcal H$ which satisfies the Jacobi identity; furthermore, such a bilinear map produces a Lie bracket on 
$\mathcal W_a\oplus  \mathcal H$ by the ansatz,
$$[X\oplus Z, Y\oplus Z] _{\mathcal V_a}= [X, Y] _{\mathcal W_a}+\omega (X,Y).$$
All such $\omega$'s are in one-to-one correspondence splittings of $\pi : V_a \to \mathcal W_a$, i.e., linear maps $\alpha : \mathcal W_a \to \mathcal V_a$ such that $\alpha \circ \pi = id _{\mathcal W_a}$. This correspondence follows the prescription
$$\omega (X,Y)=[\alpha (X), \alpha (Y)]- \alpha ([X,Y]) .$$ 
Clearly the multimodule properties of $[\phi \otimes L _p, \psi \otimes L_q]$ imply that $\omega$ must be $\star_{p,q}$-bilinear.

\end{proof}

\subsection{The Virasoro Algebra}
\begin{theorem} \label{main}
One has that 
$$H^2(\mathcal W_a,\mathcal H)\simeq \mathcal H.$$
In other words central $\mathcal H$-extensions of $\mathcal W_a$ are parametrized by $\mathcal H$. Furthermore $\mathcal V_a$ is spanned by $\{ L_m^a, Z^a\}$ with commutator relations
\begin{equation} [L_m,L_n]= A_{m,n} L_{m+n} + \delta _{m+n} h(n)\; c \;Z^a\end{equation}
where $h(n), c\in \mathcal H$ and $h(n)$ is defined by the recursive relation
\begin{equation}\label{recursive}\left\{ \begin{aligned}& h(2)=c \\&
\frac{A_{-1, -m} \Gamma _{(-(m+1))}- A_{m,1} \Gamma _{m+1}}{A_{-(m+1), m+1}}  \,  \,h((m+1)) = \frac{A_{m+1, -1} \Gamma _1 - A_{1,-(m+1)} \Gamma _{-m}}{A_{m,-m}}  \,h(m).\end{aligned}\right. 
\end{equation}
\end{theorem}

\begin{proof}
Let us analyze the meaning of condition $(3)$ (the Jacobi like identity) in the definition of $Z^2(\mathcal W_a,\mathcal H)$. By bilinearity of $\omega$, it is enough to study the relation on the generators of $\mathcal W_a$ as an $\mathcal H$-module. The Jacobi identity then reads
\begin{equation} \label{jacobi}\omega \left( L_\ell, [L_m, L_n]\right)+\omega \left( L_m [L_n, L_\ell]\right)+ \omega \left( L_n, [L_\ell, L_m]\right)=0.\end{equation}

\noindent
Using Eq. \eqref{brackets}, that is, $[L_n, L_m] \phi (z) = \sum _k a_k A_{n,m}^a(k)  \, L_{n+m} (z^k)= ( A_{n,m}^a(s) \otimes L_{n+m}) (\phi(z))$
and observing that $\star_{\mathcal H, p,q}$-bilinearity implies that
$$\omega (L_p, \psi \otimes L_q)= \frac{\psi(s+p)\Gamma_p-\psi(s) \Gamma_q}{A_{p,q}}\, \omega (L_p, L_q),$$
lead us to rewrite the  Jacobi identity as
\begin{equation} \label{jacobi2}\begin{aligned} 0&= \frac{A_{m+\ell, n+\ell} \Gamma _\ell - A_{m,n} \Gamma _{m+n}}{A_{\ell, m+n}} \omega _{\ell, m+n} +  \frac{A_{n+m, \ell+m} \Gamma _m - A_{n,\ell} \Gamma _{n+\ell}}{A_{m,n+\ell}} \omega _{m, n+\ell}\\&+ \frac{A_{\ell+n, m+n} \Gamma _n - A_{\ell,m} \Gamma _{\ell+m}}{A_{n,\ell+m}} \omega _{n,\ell+m} ,
\end{aligned}\end{equation}
where we have denoted by $\omega _{p,q}:= \omega (L_p, L_q)\in \mathcal H$. 

Setting $\ell=0$ in the equation above and using that $\omega _{m,n} =-\omega _{n,m}$
gives rise to the relation,
 \begin{equation} \label{jacobil=0}  \begin{aligned} &0=\frac{A_{m, n} \Gamma _0 - A_{m,n} \Gamma _{m+n}}{A_{0, m+n}}\, \omega _{0,m+n}\\& +\left(\frac{A_{n, m+n} \Gamma _n - A_{0,m} \Gamma _{m}}{A_{n,m}} - \frac{A_{n+m, m} \Gamma _m - A_{n,0} \Gamma _{n} }{A_{m,n}} \right) \omega_{n,m}.\end{aligned} \end{equation}
From $A_{p,q}(s)= \Gamma_p(s) -\Gamma_q(s)$ and Lemma \ref{algebraicfact} (below), we have that
$$ \frac{A_{n, m+n} \Gamma _n - A_{0,m} \Gamma _{m}}{A_{n,m}} - \frac{A_{n+m, m} \Gamma _m - A_{n,0} \Gamma _{n} }{A_{m,n}} =A_{n+m,0},$$
and that
$$\frac{A_{m, n} \Gamma _0 - A_{m,n} \Gamma _{m+n}}{A_{0, m+n}}= A_{m,n}\frac{\Gamma _0-  \Gamma _{m+n}}{A_{0, m+n}}=A_{m,n}.$$
\noindent
Thus, Eq. \eqref{jacobil=0} reduces to
\begin{equation} \label{jacobil=0-bis} \omega_{n,m} = \frac{A_{m,n}}{A_{n+m,0}}\; \omega _{0,m+n} ,\end{equation}
when $m\neq -n$.
Since an element of $H^2(\mathcal W, \mathcal H)$ is defined as an element of $Z^2(\mathcal W, \mathcal H)$ modulo co-boundaries (i.e. elements in $B^2(\mathcal W, \mathcal H)$), we can modify the representative $\omega$ of the class $[\omega ]\in H^2(\mathcal W, \mathcal H)$ by adding a co-boundary and not change the class $[\omega]$.
It is this redundancy of description that leads to central extension of the Witt algebra to form the Virasoro algebra.

We then consider the new co-cycle (still in the same class of $\omega$)
$$\omega '= \omega + \delta \mu,$$
where $\lambda$ is the linear map that  is defined (on a basis) by
$$\left\{ \begin{aligned}&\mu(L_p) = \frac{1}{A_{p,0}} \; \omega (L_0,L_p) \;\;\;\;\;\;\;\;\; \text{ for } \; p\neq 0 \\&\mu (L_0) = -\frac{1}{A_{1,-1}} \omega (L_1,L_{-1}) . \end{aligned}\right.$$
It is now straightforward to verify that, for $n+m\neq 0,$
$$\omega '(L_n, L_m) =0.$$
In fact by definition, and using Eq. \eqref{jacobil=0-bis}
$$\begin{aligned} \omega '(L_n, L_m) &= \omega (L_n, L_m) + \mu ( [L_n,L_m])=\frac{A_{m,n}}{A_{n+m,0}}\, \omega _{0,n+m}+ \mu (A_{n,m}\otimes L_{n+m})\\&= \frac{A_{m,n}}{A_{n+m,0}}\, \omega _{0,n+m}- \frac{A_{m,n}}{A_{n+m,0}}\, \omega _{0,n+m}=0. \end{aligned}$$
Therefore, $\omega '(L_n, L_m)$ is non-zero only if $n+m=0$, and thus,
 $$\omega _{m,n} '=\delta _{m+n} h_n(s).$$
 Now, using that $\mu (L_0) = -\frac{1}{A_{1,-1}} \omega (L_1,L_{-1})$, we have that
 \begin{equation}\label{h1=0} h(1)=0.\end{equation}
 In fact,
 $$\begin{aligned}&h(1)= \omega '(L_1,L_{-1}) =  \omega (L_1,L_{-1})+ \mu ([L_1,L{-1}]) = \omega (L_1,L_{-1})-  A_{1,-1} \mu (L_0)\\&=  \omega (L_1,L_{-1})- A_{1,-1} \frac{1}{A_{1,-1}} \omega (L_1,L_{-1}) =0.\end{aligned}$$
 Also, since $\omega'$ is alternate
 $$h(0)=0.$$
 We next compute $h(n)$. In order to do this, we plug in $\omega _{m,n} '=\delta _{m+n} h_n(s)$ back into Eq. \eqref{jacobi2} --which we now use for $\omega'$ instead of $\omega$-- and obtain,
 
  \begin{equation} \label{jacobiprimed}\begin{aligned}  &0=\frac{A_{m+\ell, n+\ell} \Gamma _\ell - A_{m,n} \Gamma _{m+n}}{A_{\ell, m+n}} \omega _{\ell, m+n} '+  \frac{A_{n+m, \ell+m} \Gamma _m - A_{n,\ell} \Gamma _{n+\ell}}{A_{m,n+\ell}} \omega _{m, n+\ell}'\\&+ \frac{A_{\ell+n, m+n} \Gamma _n - A_{\ell,m} \Gamma _{\ell+m}}{A_{n,\ell+m}} \omega _{n,\ell+m}'\\&=\frac{A_{m+\ell, n+\ell} \Gamma _\ell - A_{m,n} \Gamma _{m+n}}{A_{\ell, m+n}} \, \delta _{\ell+m+n} \,h(l) + \frac{A_{n+m, \ell+m} \Gamma _m - A_{n,\ell} \Gamma _{n+\ell}}{A_{m,n+\ell}} \delta _{\ell+m+n} \,h(m)\\&+ \frac{A_{\ell+n, m+n} \Gamma _n - A_{\ell,m} \Gamma _{\ell+m}}{A_{n,\ell+m}}  \delta _{\ell+m+n} \,h(n).
\end{aligned}\end{equation}
Since $ \delta _{\ell+m+n}=0$ unless $\ell+ m+n=0$, we can set $\ell= -(n+m).$ As a consequence
 $$\begin{aligned}  &0=\frac{A_{-n, -m} \Gamma _{-(n+m)}- A_{m,n} \Gamma _{m+n}}{A_{-(n+m), m+n}} \,  \,h(-(n+m)) + \frac{A_{n+m, -n} \Gamma _m - A_{n,-(n+m)} \Gamma _{-m}}{A_{m,-m}}  \,h(m)\\&+ \frac{A_{-m, m+n} \Gamma _n - A_{-(n+m),m} \Gamma _{-n}}{A_{n,-n}}  \,h(n)\\&=-\frac{A_{-n, -m} \Gamma _{(-n+m)}- A_{m,n} \Gamma _{m+n}}{A_{-(n+m), m+n}}  \,  \,h((n+m)) + \frac{A_{n+m, -n} \Gamma _m - A_{n,-(n+m)} \Gamma _{-m}}{A_{m,-m}}  \,h(m)\\&+ \frac{A_{-m, m+n} \Gamma _n - A_{-(n+m),m} \Gamma _{-n}}{A_{n,-n}}  \,h(n),
\end{aligned}$$
where we have used that $h(n)$ is odd: $h(-(n+m))=-h(n+m)$ (since $\omega '$ is alternate).
Next, we set also $n=1$ so as to obtain after using that $h(1)=0$ as per Eq. \eqref{h1=0}, the following recursive relation,
\begin{equation}\label{recursive2}
\frac{A_{-1, -m} \Gamma _{(-(m+1))}- A_{m,1} \Gamma _{m+1}}{A_{-(m+1), m+1}}  \,  \,h((m+1)) = \frac{A_{m+1, -1} \Gamma _1 - A_{1,-(m+1)} \Gamma _{-m}}{A_{m,-m}}  \,h(m),
\end{equation}
and thus $h(m)$ is completely determined by $c:=h(2)$.

\end{proof}
Next, we show a mere algebraic fact that was used in the proof above, namely
\begin{lemma}\label{algebraicfact}
$$ \frac{A_{n, m+n} \Gamma _n - A_{0,m} \Gamma _{m}}{A_{n,m}} - \frac{A_{n+m, m} \Gamma _m - A_{n,0} \Gamma _{n} }{A_{m,n}} =A_{n+m,0}$$
and
$$\frac{A_{m, n} \Gamma _0 - A_{m,n} \Gamma _{m+n}}{A_{0, m+n}}=A_{m,n}.$$
\end{lemma}

\begin{proof}
Since
$$A_{p,q}(s)= \Gamma_p(s) -\Gamma_q(s),$$
we have that
$$\begin{aligned} & \frac{A_{n, m+n} \Gamma _n - A_{0,m} \Gamma _{m}}{A_{n,m}} - \frac{A_{n+m, m} \Gamma _m - A_{n,0} \Gamma _{n} }{A_{m,n}} \\&=
\frac{\Gamma_n ^2- \Gamma_{m+n} \Gamma _n- \Gamma _0\Gamma _m+ \Gamma _m^2+ \Gamma _{n+m}\Gamma _m-\Gamma _m^2- \Gamma _n^2+\Gamma _0\Gamma _n}{A_{m,n}}\\&=\frac{\Gamma _{n+m}(\Gamma _m-\Gamma _n)- \Gamma _0(\Gamma _m-\Gamma _n)}{A_{m,n}}= \Gamma _{n+m}- \Gamma _0= A_{n+m,0}\end{aligned},$$
which proves the first identity and
$$\begin{aligned} & \frac{A_{m, n} \Gamma _0 - A_{m,n} \Gamma _{m+n}}{A_{0, m+n}}= A_{m,n}\frac{\Gamma _0-  \Gamma _{m+n}}{A_{0, m+n}}=A_{m,n}
\end{aligned}$$
proves the second.

\end{proof}
%
%
%
%
%
\section{Asymptotics of functions in $\mathcal H$}

Here we  show that for $a<1$ the functions $\Gamma _p(s)$ are sub-linear as $\vert s\vert\to +\infty$ and that $A_{p,q}(s)\to 0$ as $\vert s\vert\to +\infty$.
More precisely, we propose as follows.
\begin{proposition}
If $0<a<1$,
$$\Gamma _p(s) \sim \left( \frac{a}{e}\right) ^a (s+p)^a,$$
and
$$\lim _{\vert s\vert \to +\infty} A_{p,q}(s)=0.$$
Therefore, for any $\phi\in \mathcal H$, either
$$\phi (s)  \sim C\; (s+p)^a$$
for some constant $C$, or
$$\lim _{\vert s\vert \to +\infty} \phi (s) =0.$$
\end{proposition}
\begin{proof}
Stirling's formula dictates that
$$\Gamma(z+1) \sim \sqrt{2\pi z} \left( \frac{z}{e}\right)^z;$$
whence
$$\begin{aligned} &\Gamma _p(s) = \frac{\Gamma (a(s+p)+1)}{\Gamma(a(s+p-1)+1)} \\&\sim \sqrt{ \frac{2\pi (a(s+p))} {2\pi(a(s+p-1))}} \left( \frac{(a(s+p)) }{e}\right)^{(a(s+p))} \left( \frac{e}{(a(s+p-1))}\right)^{(a(s+p-1))}\\& =\sqrt{   \frac{s+p}{s+p-1}}  \left( \frac{(a(s+p)) }{e}\right)^a\sim\left( \frac{(a(s+p)) }{e}\right)^a, \end{aligned}
$$
thereby corroborating our first claim.
Next, using this relationship, we find that
$$A_{p,q} (s) \sim \left( \frac{a}{e}\right) ^a \left( (s+p)^a- (s+q)^a\right),$$
and all we need to show is that
$$\lim _{\vert s\vert \to +\infty} \left( (s+p)^a- (s+q)^a\right)=0.$$
In order to do that we appeal to the generalized binomial theorem which establishes that, for $\vert s\vert > |p|$,
$$ (s+p)^a= \sum _{k=0}^\infty \; {{a}\choose{k}} \, s^{a-k}p^k,$$
and likewise, for $\vert s\vert > |q|$,
$$ (s+q)^a= \sum _{k=0}^\infty \; {{a}\choose{k}}\, s^{a-k}q^k,$$
where the binomial coefficients are defines by,
$${{a}\choose{k}}= \frac{\Gamma (a+1)}{k! \,\Gamma (a-k+1)}.$$
Therefore, for  $\vert s\vert > \max\{|p|,|q|\}$,
$$(s+p)^a-(s+q)^a= \sum _{k=1}^\infty \; {{a}\choose{k}}\, s^{a-k}(p^k-q^k).$$
Since $a<1$, clearly for $k\geq 1$, $s^{a-k}\to 0$ as $\vert s\vert \to +\infty$ and hence the claim is proven.
The last statement about functions in $\mathcal H$ follows from the previous two statements and the very definition of $\mathcal H$, which is generated by the constants to be closed under the brackets, $[\phi(s), \psi(s)]_{\mathcal H, p,q}= \phi(s+p)\psi (s) \Gamma _p- \psi (q+s)\phi(s) \Gamma _q$.

\end{proof}

The main application of the study of asymptotics of elements in $\mathcal H$ is in showing the following simple structure of the algebra $\mathcal H_a$.

\begin{theorem}
$\mathcal H_a$ has a filtration $\mathcal H_0=\mathbb C\subset \mathcal H_1\subset \mathcal H_2\subset \cdots$ consisting of finite dimensional vector spaces, for $a\in \mathbb Q$.
\end{theorem}
\begin{proof}

First, we define $\mathcal H_1$ as the vector space $\bigoplus _{p,q\in \mathbb Z} A_{p,q}\mathbb C$ and $\mathcal H_2$ the vector space generated by elements of the form:
\beq \label{definitingrel} \psi (s+p) \phi (s) \Gamma _p - \phi (s+q) \psi (s) \Gamma _q,\eeq
where $\psi , \phi \in \mathcal H_1$. More generally, we define $\mathcal H_\ell$ inductively as the vector space generated by elements of the in equation \eqref{definitingrel} with $\psi , \phi \in \mathcal H_{\ell-1}$.

The main idea is to show that there are finitely many elements in $\mathcal H_\ell$ that do not have finitely many assigned poles and which generate via linear combinations all the other elements with infinitely many poles. Then the asymptotics previously discussed will do the trick.

Observe that because the Gamma function has simple poles at negative integers $\{ -k:\;\;\;k\in \mathbb N\}$ with residue at $-k$ equal to,
$$a_{-1}= \frac{(-1)^k}{k!},$$
analogously the functions $\Gamma (a(s+p)+1)$ and $\Gamma(a(s+p-1)+1)$, which appear in the ratio $\Gamma_p(s)= \frac{\Gamma (a(s+p)+1)}{\Gamma(a(s+p-1)+1)}$, also have poles of the same type at the points $s_{k,p}= \frac{-k}{a}-p$ and $t_{h,p}= \frac{-h}{a}-p+1$ respectively, with $k,h\in \mathbb N\setminus \{0\}$. One readily verifies that the functions $\Gamma \left(a(s+p)+1\right)$ and $\Gamma \left(a(s+p-1)+1\right)$ never share the poles $s_{k,p}$ and $t_{h,p}$. 

We then concentrate on the difference $A_{p,q}=\Gamma_ p -\Gamma _q$. We assume that $a\in \mathbb Q$, so we can write $a=\frac{M}{N}$ with $(M,N)=1$ (i.e., they are coprime). Since $\Gamma _p$ and $\Gamma _q$ have the same residues at poles, if we show that for some pole $s_{k,p}$ of $\Gamma _p$ there is a corresponding pole $s_{q,h}$ for $\Gamma _q$  such that $s_{k,p}= s_{q,h}$, then at such point $A_{p,q}$ will be regular.
One readily verifies that
$$s_{k,p}= s_{h,q}$$
if and only if
$$N(k-h)= M(q-p)$$
which is solvable if and only if $q-p$ is divisible by $N$. Thus, the $A_{p,q}$ which do not have finitely many poles are the ones for which $p\not\equiv q (\rm mod \; N)$.
Next, observe that if $p\not\equiv q (\rm mod \; N)$, we can write:
$$A_{p,q}= A_{p,n}+A_{n,q}$$
and we can choose $n$ so that $n\equiv q (mod. \; N)$, thus reducing to the case where we only have finitely many $A_{p,q}$ for any given $p$. But by the symmetry of $A_{p,q}$, we can apply the same argument to $p$ and reduce ourselves to essential $A_{p',q'}$ with $0\leq p',q'<N$. This shows that $\mathcal H_1$ is generated by the finitely many $A_{p',q'}$ with $0\leq p',q'<N$ and a set of $A_{p,q}$'s with $p\equiv q (mod. \; N)$. The latter $A_{p,q}$'s span a finitely dimensional vector space due to their sublinear asymptotics and the fact that they all have fintely many poles.
The arguments for the other $\mathcal H_k$'s are similar.
\end{proof}
%

\subsection{Final Remarks}

We have constructed a Virasoro algebra with non-local operators as the generators of the currents.  This algebra should be relevant to any problem in which the effective currents are non-local.  The resultant algebra reflects the non-locality as it cannot be defined independent of the basis.  Consequently, the central charge is explicitly a function rather than a number.  Evaluated on the basis results in a charge for every degree, where the degree is the exponent of the element, $z^p$. While the next step is to construct explicitly the representations that generate the fractional currents, this algebra opens the possibility of new models in string theory based on such non-local currents as the basis for the conformal sector.  

\section*{Acknowledgements}
We thank the NSF DMR-1461952 for partial funding of this project.  PWP thanks David Lowe for a useful conversation in which he insisted that the urgent problem which can put anomalous gauge fields on a firm theoretical footing is the construction of the associated Virasoro algebra.


\end{document}